\documentclass[english,11pt]{article}
\usepackage{custom_tex}

\begin{document}
\title{FACT: Fast closed testing for exchangeable local tests}

\date{\today}
\author{Edgar Dobriban\footnote{Department of Statistics, The Wharton School, University of Pennsylvania. E-mail: \texttt{dobriban@wharton.upenn.edu}.}  }

\maketitle

\graphicspath{
{../Experiments/Experiment 1 - Fisher vs Bonferroni/}
{../Experiments/Experiment 2 - Additive/}
{../Experiments/Experiment 3 - Mixture/}
{../Experiments/Experiment 3b - Mixture Ada/}
{../Experiments/Experiment 3c - Mixture Corr/}
{../Experiments/Experiment 4 - C4D/}
}

\abstract{
Multiple hypothesis testing problems arise naturally in science. In this paper, we introduce the new Fast Closed Testing (FACT) method for multiple testing, controlling the family-wise error rate. This error rate is state of the art in many important application areas, and is preferred to false discovery rate control for many reasons, including that it leads to stronger reproducibility. The closure principle rejects an individual hypothesis if all global nulls of subsets containing it are rejected using some test statistics. It takes exponential time in the worst case. When the tests are symmetric and monotone, our method is an exact algorithm for computing the closure, quadratic in the number of tests, and linear in the number of discoveries. Our framework generalizes most examples of closed testing such as Holm's and the Bonferroni method. 
As a special case of our method, we propose the Simes-higher criticism fusion test, which is powerful for detecting both a few strong signals, and also many moderate signals.}

\section{Introduction}

We often need to make multiple decisions, while ensuring that we do not make many mistakes. In the model problem of multiple testing, we have several hypotheses, for instance about the association between explanatory variables and an observed outcome. Based on the observed data, we want to discover the variables truly associated with the outcome, while controlling the erroneous discoveries. 

The classical error criterion to control in multiple testing is the family-wise error rate, which is the probability of any erroneous discoveries.
Family-wise error rate control requires that the probability of making any false positives is small, which is very appealing in critical application areas where mistakes are expensive.
While false discovery rate \citep{benjamini1995controlling} control is sometimes preferred, family-wise error rate control
remains the standard in high-stakes areas where more rigorous error control is desired, such as in 
end-stage
genome-wide association studies \citep{sham2014statistical}, neuroimaging \citep{eklund2016cluster}, as well as medical and pharmaceutical applications \citep{dmitrienko2009multiple}. 

The closure principle is a general framework for constructing multiple tests controlling the family-wise error rate \citep{marcus1976closed}. Suppose that we have individual null hypotheses $H_1,\ldots, H_n$, and that we want to control the family-wise error rate at level $\alpha$. For instance, $H_i$ may correspond to the null hypothesis that the $i$th explanatory variable is not associated with the outcome. 

In the closure principle, we start with valid level $\alpha$ local testing rules for all nulls that are intersections of subsets of individual null hypotheses. An intersection null $H_J$ states that all of the hypotheses $H_i$, $i \in J$ are null. In closed testing, we reject the individual null $H_i$ if, for all sets $J$ such that $i\in J$, the intersection null $H_J$ is rejected. To argue that we have an individual discovery, we must be able to see this discovery in all groups of individual hypotheses. This method controls family-wise error rate strongly, under any configuration of true and false individual null hypotheses \citep{marcus1976closed}.

Clearly, this method is computationally intractable in general, because it requires testing all $2^n-1$ intersection nulls, corresponding to all nonempty subsets of hypotheses. In practice, popular computationally efficient shortcuts are used in special cases. For instance, the important Holm's procedure \citep{holm1979simple}, used in tens of thousands of scientific studies, is the closure of Bonferroni's method, and Hommel's procedure \citep{hommel1988stagewise} is the closure of the Simes method  \citep{simes1986improved}. Hochberg's procedure \cite{hochberg1988sharper} is a simple method that is strictly more conservative than Hommel's procedure. Liu's method \citep{liu1996multiple} gives conditions under which the closure principle can be reduced to step-up or step-down methods.

However, these uses of the closed testing principle are limited, and they do not allow much flexibility. It would be desirable to choose local testing rules that can use information adaptively. If there are only few large effects, we should use a local testing rule that only looks at the smallest $p$-values, for instance Bonferroni's or the Simes method. If there are many moderately small effects, then we should use a local testing rule that uses many of the small $p$-values, like the chi-squared or higher criticism tests. There are currently no known computationally efficient methods for this. 

To broaden the methods available for multiple testing, in this paper we introduce Fast Closed Testing (FACT), a new computationally efficient framework for multiple testing based on the closure principle. Our algorithm allows the user to design powerful multiple testing architectures for specific applications. The key requirement on the problems where it can be used is the exchangeability of the hypotheses: we treat the hypotheses as if there was little or no distinction between them. This is a reasonable first order approximation in many applications. The use of prior information can be difficult to quantify and can bias results; see below for more discussion. Therefore, symmetry is often a good assumption even in the presence of some prior information. 

The closed testing principle is more than 40 years old. While there is a lot of related work, see Section \ref{rw2}, it appears that the FACT algorithm has not been reported before. We found this surprising, and closing this gap was our main motivation. 

\section{Fast closed testing}

 More formally, let $p_i$, $i=1,\ldots,n$ be the $p$-values for the individual null hypotheses $H_i$. The $p$-values are assumed to be uniform, or stochastically larger than uniform random variables whenever the null hypotheses $H_i$ are true. Depending on the context, we will need specific assumptions on the joint distribution of the $p$-values. These will range from no additional assumptions to joint independence.

Denote by $p_J$ the vector of $p$-values with indices in the set $J \subset \{1,\ldots,n\}$, and consider a local testing rule for the intersection null 
 $$H_J = \bigcap_{i\in J} H_i,$$ 
 with decision function 
 $$\Phi(p_J):[0,1]^J\to\{0,1\}.$$ The intersection null $H_J$ is rejected based on the $p$-values $p_J$ if $\Phi(p_J)=1$.

The closed testing method \citep{marcus1976closed} defines a multiple testing procedure for the individual null hypotheses $H_i$ based on the local testing rules. In words, we reject $H_i$ if, for all sets $J$ such that $i\in J$, the null $H_J$ is rejected. Formally, the closed testing decision rule $\Phi_c$ for the $i$th null is

$$\Phi_c(p_i) = \prod_{J: i\in J}\Phi(p_J).$$
Some of the null hypotheses $H_i$ are true, while some are false, but we do not know which ones. An intersection null $H_J$ is true if and only if all of the individual hypotheses $H_i$ with $i \in J$ are true. The local testing rules have level $\alpha$ if the probability of rejecting $H_J$ when it is true is at most $\alpha$. The  family-wise error rate of a multiple testing rule is the probability of rejecting any individual null hypothesis $H_i$ when it is true. The closed testing method controls the family-wise error rate, if each local testing rule has level $\alpha$, see the Supplementary Material for the well-known argument. 

To implement closed testing, we need to find the indices $i$ for which all parent subsets are rejected. In general, we would need to check each of the $2^{n-1}$ subsets containing an index. However, this can simplify under some conditions, as there can be many sets that we do not need to check. 

We assume that the problem has a global symmetry structure, meaning that we treat the hypotheses identically. In that case, it makes sense to use the same local testing rule $T_j$ for each subset of a fixed size $j$. Suppose also that the local testing rules satisfy the following two conditions:

\begin{assumption}[Monotonicity] If we reject the intersection null and decrease any $p$-value, we still reject it. If for all coordinates $i$, $p_i \le p'_i$, then 
$$\Phi (p_J) \ge \Phi (p'_J).$$
\end{assumption}

\begin{assumption}[Symmetry] Rejection only depends on the set of $p$-values, and not on the labels of their null indices. For any permutation $\pi$ of the indices $J$, 
$$\Phi (p_J) = \Phi\left\{p_{\pi(J)}\right\}.$$
\end{assumption}

Under these conditions only the sizes of the $p$-values, and not their indices, matter to choose the rejections. Moreover, smaller $p$-values are always better. This implies that some group of the smallest $p$-values will be rejected. 

\begin{algorithm}[!h]
\caption{Fast Closed Testing (FACT)} \label{fct}
\begin{tabbing}
   \enspace Input $p$-values $p_i$, $i = 1, \ldots, n$ \\
   \enspace Input local testing rules $T_k$ for testing subsets of size $k$\\ 
   \enspace Input desired family-wise error rate $\alpha$\\
   \enspace Sort the $p$-values: $p_{(1)} \le \ldots \le p_{(n)}$\\
   \enspace For $k = 1$ to $k=n$ \\
   \qquad For $j = k$ to $j=n$\\
  \qquad\qquad  If the null with $p$-values $p_{(k)},p_{(j+1)},p_{(j+2)},\ldots,p_{(n)}$ is not rejected\\ 
  \qquad\qquad using local testing rule $T_{n-j+1}$ go to next line 
  \\
\enspace Reject hypotheses corresponding to the $p$-values $p_{(1)},\ldots,p_{(k-1)}$
\end{tabbing}
\vspace*{-10pt}
\end{algorithm}

It remains to find this group. Intuitively, we should start with the smallest $p$-value, and check if the hardest subset of every given size containing it will be rejected. We we give an equivalent sequential algorithm relying on two  loops. We call this the Fast Closed Testing (FACT) algorithm, see Algorithm \ref{fct}.

In this algorithm, we make the following conventions. The indices in the loops run until $n+1$, but the $p$-value  $p_{(n+1)}$ is taken to be the empty set. This element is not included in the list of $p$-values to be tested. Moreover, the $p$-values $p_{(j+1)},\ldots,p_{(n)}$ are also taken to be the empty set when $j=n$. In that case they are not included in the list. Finally, when $k=n+1$ and $j=n$, the above conventions specify that the $p$-values to be tested are the empty set, so the list is empty. In this case, we make the convention that this set is always rejected. 

Our main result shows that the FACT algorithm controls the family-wise error rate, and characterizes the running time of the algorithm. See the Supplementary Material for the proof.

\begin{theorem}[Correctness of the FACT algorithm]
\label{thm_fct}
If the intersection null hypotheses are tested with monotone symmetric testing rules $T_j$ with level $\alpha$, using the same rule for all sets with the same size $j$, then the FACT algorithm provides exactly the same rejections as closed testing. Hence it controls the family-wise error rate at level $\alpha$. 

Moreover, suppose that applying the local testing rule $T_j$ to $j$ hypotheses takes linear time $O(j)$. Then, the FACT algorithm takes $O(sn^2)$ time, where $s$ is the number of significant discoveries, or rejections.
\end{theorem}

The significance of the FACT method is that it  enables using the closed testing framework in a more flexible way than what was known before. It is possible to design closed testing architectures based on the FACT algorithm for specific applications. Before the FACT method, the applications of  closed testing were limited to either very special cases like Holm's method, or were computationally intractable. We show in the Supplementary Material that the FACT algorithm generalizes well-known shortcuts for closed testing. %
The intuition behind using symmetric rules is that hypotheses are exchangeable. Then, for every subset of $p$-values we need to think about a test statistic that has a good chance to detect any signal in that group. By choosing the local testing rules appropriately for specific applications, practitioners may design powerful new closed testing architectures. In the next section we will illustrate the steps.
Later we will give a broader set of examples. 

\begin{algorithm}[!h]
\caption{Adjusted $p$-values for FACT} \label{fct:a}
\begin{tabbing}
   \enspace Input $p$-values $p_i$, $i = 1, \ldots, n$ \\
   \enspace Input local testing rules $T_k$ for testing subsets of size $k$\\ 
   \enspace Sort the $p$-values: $p_{(1)} \le \ldots \le p_{(n)}$\\
   \enspace For $k = 1$ to $k=n$ \\
   \qquad For $j = 1$ to $j=n$\\
  \qquad\qquad  Let $q_{kj}$ be the $p$-value for the null with $p$-values $p_{(k)},p_{(j+1)},p_{(j+2)},\ldots,p_{(n)}$\\
  \qquad\qquad using local testing rule $T_{n-j+1}$ \\
  \qquad Let $\tilde p_{(k)} = \max_{j} q_{kj}$ \\
\enspace Output adjusted $p$-values $\tilde p_{(1)},\ldots,\tilde p_{(n)}$
\end{tabbing}
\vspace*{-10pt}
\end{algorithm}

It is often of interest to compute adjusted $p$-values for a multiple testing procedure. An adjusted $p$-value of a hypothesis $H_i$ is the smallest critical value at which the hypothesis is rejected. To complement FACT, we also show how to compute adjusted $p$-values for the method in Algorithm \ref{fct:a}. This algorithm requires as input methods for evaluating $p$-values for the local testing rules $T_k$. Its computational complexity is $O(n^3)$ for testing rules for which finding critical values takes linear time. While computing the adjusted $p$-values is based on the same idea as the FACT algorithm itself, we present this method separately, because in principle it can be run independently from FACT, and because it takes time $O(n^3)$ instead of $O(sn^2)$.

\subsection{Simes-higher criticism fusion rule}
\label{sec_ex2}
The flexibility of the FACT algorithm is achieved by using different local testing rules for intersection nulls of different subset sizes. We will show here that this enables the design of tests that can be powerful against both sparse and dense alternatives. 

Suppose that we are in a setting where we have $n$ total $p$-values and the model is $s$-sparse, in the sense that there are $s$ false nulls; or equivalently nonzero effect sizes. If we are testing an intersection null of size $j$ containing a specific false null, then in the worst case there are 
$$A = \max\{j-(n-s),1\}$$ 
false nulls in this set. Now, the closure principle must reject each subset of size $j$ containing the hypothesis in order to reject the individual null. Therefore, to maximize the chances that we reject this false null, we should use only the smallest $A$ $p$-values out of the total of $j$ $p$-values in the test statistic $T_j$. 

The above reasoning gives a heuristic for the maximal number of non-nulls that each local testing rule can use. When testing subsets of a small size, we should use only a small number, and possibly only one $p$-value. This suggests that we should use the Bonferroni or the Simes rules for small subsets. Specifically, given a guess $s$ for the sparsity, in this section we propose to use the Simes rule for subsets of size at most $n-s+1$. The reason why we use Simes is that it is strictly more powerful than the Bonferroni method, and it is known to work under general correlation structures, see the Supplementary Material for details.

 Given an intersection null of size $j$, and $p$-values $p_i$ for this null, Simes rule sorts them, and rejects the intersection null if, for any $p$-value, $p_{(i)}/i $ is at most $\alpha/j$ \citep{simes1986improved}. Formally, it rejects if

$$\min_i \frac{p_{(i)}}{i} \le \frac{\alpha}{j}.$$

The Simes test is more powerful than the Bonferroni test, and has correct level under a broad class of positive dependence structures, see for instance \cite{goeman2014multiple,tamhane2018advances} for details. We will see later that both the Bonferroni and Simes rules are symmetric and monotone, so they can be used with FACT, see the Supplementary Material.

When the subset size is large, we should use local testing rules that are powerful against relatively denser alternatives. Given that the higher criticism rule for local testing is known to be effective against many types of alternatives \citep{donoho2004higher}, we propose to use it for large subsets. The higher criticism for testing a global null based on $j$ $p$-values works as follows.  For a fixed critical value $\beta$, under the global null, the number of $p$-values less than $\beta$ follows a binomial distribution with $j$ trials and success probability $\beta$. Therefore, we can compute the fraction $f_\beta$ of $p$-values below $\beta$, standardize it by its standard deviation, and obtain the test statistic
$$C(\beta) = j^{1/2} \frac{f_\beta-\beta}{\{\beta(1-\beta)\}^{1/2}}.$$
If this test statistic is large, then the fraction of $p$-values below $\beta$ is large, suggesting evidence against the global null. In order to be adaptive to the number of nonzero effects, the higher criticism test takes the largest of these statistics over a range of $\beta$s. This ensures that we can detect both a few large effect sizes, as well as a larger number of moderate effect sizes, see \cite{donoho2004higher} and the Supplementary Material. We will see that the higher criticism rule is symmetric and monotone, and so it can be used with FACT, see the Supplementary Material.

This leads to the  Simes-higher criticism fusion rule, which is summarized in Algorithm \ref{fct2}.

\begin{algorithm}[!h]
\caption{Simes-higher criticism fusion for FACT} \label{fct2}
\begin{tabbing}
   \enspace Input number of hypotheses $n$ \\
   \enspace Input preliminary estimate for sparsity $s$\\ 
   \enspace For $j\le n-s+1$, choose $T_j$ to be the Simes rule\\
   \enspace For $j> n-s+1$, choose $T_j$ to be the higher criticism rule\\
   \enspace Run FACT with this choice of local testing rules\\
\end{tabbing}
\vspace*{-20pt}
\end{algorithm}

We illustrate this algorithm in simulations and a data analysis example in the Supplementary Material. There we also attempt to address concerns about the practicality of this method, for instance about the assumption of known sparsity levels.

\subsection{Some related work}
\label{rw2}

Here we review some of the most closely related work in the literature. For broader reviews of multiple testing, see for instance \cite{hochberg1987multiple, goeman2014multiple, bretz2016multiple}. See also \cite{tamhane2018advances} for a review of $p$-value based methods focusing on more recent methods.
The relationship between closure and computationally efficient algorithms has been studied from various perspectives. For instance, \cite{grechanovsky1999closed} establish conditions when a closure has  a step-down sequentially rejective shortcut.  Our work has a broader scope, because step-down procedures are a very special case of efficient algorithms. The FACT method is often not a step-down algorithm. \cite{gou2014class} propose improved hybrid Hochberg-Hommel type step-up multiple test procedures, which are also different from the FACT Simes-higher criticism hybrid. 

A different line of work aims to develop interpretable closed testing methods using graphical approaches \citep{bretz2009graphical}. This is important because these methods can be easily explained to practitioners. However, from a methodological and computational point of view, it is also important to develop new powerful methods such as those in our work. 

Monotonicity ideas have appeared in the literature on multiple testing.  A related monotonicity condition appeared in \cite{birnbaum1954combining}, however, it was  used for a completely different purpose than in our work. Indeed, there it was used as a condition under which meta-analysis methods are optimal.  A monotonicity of the resulting closed test has been discussed in \cite{dmitrienko2009multiple}, Section 2.3.4. However, this is a global monotonicity condition, different from ours, because it applies to the overall multiple testing procedure, as opposed to the local tests. In particular, our methods are always monotone in the global sense. \cite{hommel2008aesthetics} discuss several monotonicity requirements for tests, and mention the present notion too, see their Section 3.2. However, we also develop explicit algorithms based on this condition.

Monotonicity has also appeared as a condition for error control in the sequential testing principle of \cite{goeman2010sequential}. However, the algorithms presented there are even more general than the closed testing principle, and thus not always computable in polynomial time.  A similar observation about computationally efficient closed tests was made by \cite{henning2015closed}, see their Section 3, who also noticed that for deciding whether or not to reject $H_i$, one must identify  the hurdle subset for each subset  size. However, it appears that they did not explicitly describe an algorithm to do so in full detail.

Studying a different problem, that of constructing a confidence statement on the number of false rejections incurred, \cite{goeman2011multiple} also construct shortcuts for exchangeable local tests, such as the Fisher's test, see their Section 4 and Appendix A. While  they are based on the same principle, looking at the worst case set at each level, the two algorithms are different.

\section{Constructing monotone symmetric rules}
\label{sec_con}

How should we construct monotone symmetric rules? In this section we discuss some general principles. Suppose we use a test statistic based rule, where we compute some test statistic $T = T(p) $ based on the $p$-values, and reject if this test statistic is less than some critical value $ c_{\alpha}$. Formally, the rejection rule has the form
 $$\Phi(p) = I\{T(p) \le c_{\alpha}\}.$$ 
 
 When does such a test become monotone and symmetric? It is easy to see that this will hold if the test statistic $T$ itself is also monotone increasing and symmetric, in the following sense: 

\begin{assumption}[Monotonicity] If we decrease any $p$-value, the value of the test statistic $T$ decreases.  Formally, $T(p) \le T(p')$ if on all coordinates $i$, $p_i \le p'_i$. 
\end{assumption}

\begin{assumption}[Symmetry] The test statistic $T$ only depends on the set of $p$-values, and not on their indices. Formally, $T(p) = T(p_{\pi})$ for any permutation $\pi$. 
\end{assumption}

Moving from rejection rules to test statistics is valuable, because one can naturally combine them as below. See the Supplementary Material for the proof.

\begin{lemma}[Constructing monotone symmetric test statistics]
\label{cl_lem}
Let $T^i, i \in I$ be any collection of monotone symmetric test statistics. Then, one can construct new monotone symmetric test statistics by taking: 
\begin{property}
Minima: $\min_{i \in I} T^i$ is monotone symmetric.
\end{property}

\begin{property}
Maxima: $\max_{i \in I} T^i$ is monotone symmetric. 
\end{property}

\begin{property}
Non-negative linear combinations: $\sum_{i = 1}^k \lambda_i T^i$  is monotone symmetric for any $\lambda_i \ge 0$. 
\end{property}

\begin{property}
Monotone functions: 
$g(T^1,\ldots,T^k)$  is monotone symmetric if the function $g$ is coordinate-wise monotone, in the sense that $g(x_1,\ldots,x_k) \le g(x'_1,\ldots,x'_k)$ if on all coordinates $i$, $x_i \le x'_i$. 
\end{property}

\end{lemma}
An important example are the order statistics.

We mention a few examples of monotone symmetric rules, deferring the details to the Supplementary Material. 
The Bonferroni type rule with $T(p_J) = \min_{j\in J}p_j$ is monotone and symmetric. In this case, as expected, the FACT algorithm can be simplified into Holm's procedure. However, Holm's procedure is more direct and has a smaller complexity of $O(n\log n)$. Similarly the Simes rule is monotone and symmetric, and the FACT algorithm can be shown to simplify into Hommel's procedure. As before, Hommel's procedure is more direct and has a smaller complexity of $O(n^2)$. Recently its complexity has been reduced to $O(n)$ for sorted $p$-values \citep{meijer2017shortcut}. Moreover, any Generalized Simes Test \citep{grechanovsky1999closed} is mononotone and symmetric. So are monotone functions of order statistics, including Fisher's and Stouffer's combination, rank-sum type statistics, the higher criticism, and the hybrid Hochberg-Hommel method \citep{gou2014class}.

Here we studied the symmetry of the test statistic. However, the symmetry of the tests also depends on the distribution of $p$-values. Here we implicitly assumed that we have exchangeable $p$-value distributions, so that we can choose the same critical value for each fixed subset size. An alternative is to set critical values using probability inequalities, such as the Bonferroni or Simes inequalities. However, these are quite rare, so this may limit the applicability of our method.

\section*{Acknowledgement}
We thank the associate editor and the referee for the helpful suggestions that have significantly improved the paper. We are grateful to Vladimir Vovk for pointing out an indexing error in the calculation of the adjusted $p$-values \citep{vovk2019combining}. We also thank Emmanuel Candes, Jelle Goeman, Jiangtao Gou,  William Leeb, Ajit Tamhane, Jingshu Wang, Min Xu,  and Nancy Zhang for valuable comments.

\appendix

\section{Examples of monotone symmetric test statistics}
\label{sec_ex_app}

\subsection{Bonferroni-type rules}
\label{br}

The simplest class of monotone symmetric rules is the  Bonferroni-type rules: 

$$T(p_J) = \min_{j\in J} p_j.$$

The critical values are  $c_{|J|,\alpha} = \alpha/|J|$.  The Bonferroni rule controls the type I error under any dependence structure. 
It is well known that the closure principle applied to the Bonferroni test becomes Holm's procedure \citep{holm1979simple}. 

Moreover, we show that after some work, the FACT algorithm reduces to Holm's procedure. See Section \ref{pf:fct_bonf} for the proof. This result is important because it shows that our FACT algorithm is sensible, as it at least matches the most well-known example of an efficient closed testing algorithm. Otherwise, this result just a basic sanity check, as we already know that both FACT and Holm's method are exact shortcuts for the closure principle, so they should lead to the same rejections. However, we still think it is illuminating to understand more deeply the connection between the two shortcuts, and that is the reason why we present this result and its proof. 

\begin{proposition}
\label{fct_bonf}
The FACT algorithm for closing the Bonferroni method leads to the same rejections as Holm's procedure. 
\end{proposition}

  

\subsection{Simes-type rules}
\label{sr}

A second class of monotone symmetric rules are  Simes-type rules, which are based on comparing the order statistics of the $p$-values to increasing thresholds:
$$\Phi(p_J) = I\left(\min_i \frac{p_{(i)}}{i} \le \frac{\alpha}{|J|}\right).$$

Therefore, $T(p_J) = \min_{i\in J} p_{(i)}/i$, and $c_{|J|,\alpha} = \alpha/|J|$ \citep{simes1986improved}. The Simes test is more powerful than the Bonferroni test. Moreover, Simes' test has exact level $\alpha$  for the intersection null under independence \citep{simes1986improved}. Under many types of positive dependence, Simes is conservative \citep{samuel1996simes,sarkar1997simes}, and the known situations in which it is anti-conservative occur under quite pathological negative dependence structures \citep{rodland2006simes}. See also \cite{goeman2014multiple, tamhane2018advances}.

As a first remark, we observe that Simes' method is symmetric and monotone by inspection. 
Next we apply closed testing to the Simes test. 
\cite{hommel1988stagewise} gives the following algorithm for the closure of Simes: Let $j$ be the largest index such that 
$$p_{(n-j+k)} > k\alpha/j$$
 for all $k=1,\ldots,j$. If $j$ does not exist, reject all $H_i$. Otherwise, reject all $H_i$ with $p_i \le \alpha/j$. This algorithm takes $O(n^2)$ in the worst case.

The FACT method takes $O(n^3)$ in the worst case, and thus clearly its steps do not agree with Hommel's method. However, we now show that there is a simplification, and the FACT method can be reduced to Hommel's procedure. As above, since since both methods are exact shortcuts, we already know that this result must be true. However, we still find it insightful to understand the connection between the methods.

\begin{proposition}
\label{fct_simes}
The FACT algorithm for closing Simes' method leads to the same rejections as Hommel's procedure. 
\end{proposition}

 See Section \ref{pf:fct_simes} for the proof. 

More recently,  \cite{meijer2017shortcut} gave a more efficient, $O(n\log n)$ algorithm for closing Simes' method. This is faster than FACT, but is limited to Simes' method. In contrast, FACT is applicable to any monotone symmetric local testing rule, including Bonferroni and monotone combinations (see below). 

Another important class of monotone local testing rules comes from the  Generalized Simes Test, \citep{grechanovsky1999closed}, which tests a intersection null $H_J$ in the following way. Let $$d_{|J|,1},\ldots, d_{|J|,1}$$ be a sequence of critical values. Reject $H_J$ if, there is $j\in J$ with $$p_{(j)} \le d_{|J|,j}.$$ Here $p_{(j)}$ refers to the ordering of the $p$-values within the set $J$. Note that for Simes test, $d_{|J|,j} = \alpha j/|J|$. \cite{liu1996multiple} gave conditions for reducibility of a closure based on this method to general sequentially rejective step-down or step-up procedures. Here we consider a more general set of such procedures.

We next show that the Generalized Simes Test is symmetric and monotone. Therefore, we can use it as a component in the FACT algorithm. 
This paves the way to a wide variety of new methods for closed testing. See Section \ref{pf:gst} for the proof. 

\begin{proposition}[GST-Sym-Mon]
\label{gst}
The Generalized Simes Test is symmetric and monotone. 
\end{proposition}

\subsection{Monotone sums and functions of order statistics}
\label{br2}

Another broad class of monotone symmetric rules is the set of  monotone sums, in which the test statistics are sums of monotone functions of the $p$-values. These are especially appealing if the effect sizes are ``dense'', in the sense that we expect to have many nonzero effects. Monotone sums include the following test statistics: 

\benum

\item {\bf  Fisher's combination}: $$T(p_J) = 2\sum_{i\in J} \ln(p_i).$$ Fisher's test has a $-\chi^2_{2k}$ distribution under the intersection null when all $p$-values are uniform and independent \citep{fisher1970statistical}. Thus. $c_{k,\alpha} = -\chi^2_{2k}(1-\alpha)$, the $100(1-\alpha)$-th percentile of the $\chi^2_{2k}$ distribution.

Fisher's combination is monotone, because the function $2\ln(x)$ is monotone increasing for $x\in[0,1]$.

\item {\bf  Stouffer's combination}: $$T(p_J) = \sum_{i\in J} \Phi^{-1}(p_i),$$ where $\Phi$ is the standard normal cdf \citep{stouffer1949american}. Stouffer's combination has a $\N(0,k)$ distribution under the intersection null when all $p$-values are uniform and independent. Thus, $c_{k,\alpha} = k^{1/2} \Phi^{-1}(\alpha)$. 

Stouffer's combination test is symmetric and monotone, for the same reasons as Fisher's test. 

\item {\bf  Wilkinson's combination}: $$T(p_J) =-\sum_{i\in J} I(p_i \le d),$$ where $d>0$ is some constant \citep{wilkinson1951statistical}. Wilkinson's combination has a sign-flipped Binomial distribution with $n$ trials and success probability $d$ under the intersection null when all $p$-values are uniform and independent. Thus, its critical values can be found from the distribution of the Binomial. 

\item {\bf  Truncated Product Method}: $$T(p_J) = \sum_{i\in J}\ln(p_i) I(p_i \le \tau),$$ where $\tau>0$ is some constant \citep{zaykin2002truncated}. The value $\tau = \alpha$ is suggested as a default \citep{zaykin2002truncated}. The critical values of this test can be found numerically.

\item {\bf  Romano-Shaikh-Wolf combination}: Suppose we observe independent random variables $X_i \sim \N(\mu_i,1)$ and we wish to test $ H_i:\mu_i=0$ against $\mu_i \neq 0$. \cite{romano2011consonance} study test statistics of the form 
$$T(X_J) = \sum_{i\in J}\cosh(\ep |X_i|)$$ for $\ep>0$. They show that the closure of these test statistics has a maximin optimality property against subsets of the alternative of the form $\gamma(\ep) = \{\mu:|\mu_i|\ge \ep, \textnormal{all } i\}$, for large enough $\ep$. This follows from their more general result that maximin optimality is inherited under closure if the resulting multiple test is consonant. Clearly, their test statistics are monotone in $|X_i|$, so they fit in our framework. 

\item {\bf  Monotone combination}: More generally, we can use sums a of monotone increasing function $f$ of the $p$-values:
$$T(p_J) = \sum_{i\in J}f(p_i).$$ 

The above tests are special cases. The critical values of this test can be found numerically. 

Any monotone combination tests takes $T_n  = O(n)$ to apply to $n$ $p$-values. Thus, the overall running time of the closed testing method is  $O(n^2k)$ if $k$ nulls are rejected, and $O(n^3)$ in the worst case.

\eenum

The power and flexibility of our method is showcased by the ability to use local testing rules that go beyond the classical ones (such as Bonferroni and Simes).  Indeed, we can use any  monotone functions of the order statistics. We give several examples below: 
\benum

\item  Rank-sum type statistics. We can use test statistics of the form $T(p) = \sum_{i}f_i(p_{(i)}),$ where $f_i$ are for monotone increasing functions, possibly changing with $i$. The Bonferroni method is a special case, where $f_1(x) = x$, and the other functions are zero. Broader examples of rank-sum type statistics include
 linear weighted rank-sum statistics. For any weights $w_i \ge 0$, we can use the linear weighted rank-sum statistics $T(p) = \sum_{i}w_i \cdot p_{(i)}.$
For instance, if we want to emphasize not just the smallest, but also the second smallest $p$-value, we may use $p_{(1)}+ \ep p_{(2)}$.

\item  Min/max type statistics.  We can also use test statistics of the form  $T(p) = \max_{i}f_i(p_{(i)})$ and $T(p) = \min_{i}f_i(p_{(i)})$ where $f_i$ are monotone increasing functions, possibly changing with $i$. The following test statistics are examples: 

\benum
\item  Generalized Simes Test. Recall that this rejects $H_J$ if there is $j\in J$ with $p_{(j)} \le d_{|J|,j}$ \citep{grechanovsky1999closed}. This falls in the max-category, where $f_i(x)=I(x>d_{|J|,i})$. 

\item  Higher Criticism: The higher criticismtest statistic was introduced by Tukey in the 1960s, and experienced a resurgence of interest after its study by \cite{donoho2004higher}. The local test can be described using the functions
$$g_i(x) =\sqrt{n} \frac{x-i/n}{\sqrt{x(1-x)}},$$
and the test statistic equals, for some $0 <\alpha_0 <1$,  $T(p) = \min_{i \le \alpha_0 n} g_i(p_{(i)}).$
The critical value for the test can be chosen as $-\sqrt{2\log\log n}(1+o(1))$ \citep{donoho2004higher}, but this may need some adjustments in finite samples. This test falls in the min-category, where $f_i=g_i$ for $i \le \alpha_0 n$, and $f_i = 0$ otherwise. The higher criticism has originally been studied under independence, but there are extensions allowing some degree of dependence. 

\eenum

A challenge with these general monotone functions is that the critical values are typically not available in closed form. However they can usually be evaluated numerically.

\item  Hybrid Hochberg-Hommel.  \cite{gou2014class} proposed improved hybrid Hochberg--Hommel type step-up multiple test procedures, and showed that they are the closures of the following local tests. We reject the intersection null, i.e., $T(p) = 1$ if one of
the following mutually exclusive events occurs:
\begin{equation*}
E_i 
= \begin{cases}
p_{(n)} \le \alpha & i=1,\\
p_{(n)} > \alpha, p_{(n-1)} > c_2\alpha,\ldots, p_{(n-i+2)} > c_{i-1}\alpha &\\
p_{(n-i+1)} \le c_{i}\alpha, p_{(1)} \le d_{i}\alpha & i\ge 2.
\end{cases}
\end{equation*}
Here $c_i,d_i$ are two monotone decreasing sequences of critical constants with $1 \ge c_i \ge d_i$. It is easy to see that this local test is monotone and symmetric. Indeed, we only need to observe that if we are in $E_i$, and decrease any $p$-value, we will either stay in $E_i$, or move to some $E_k$ with a smaller index $k <i$. 

\eenum

\section{Consonance}
\label{sec_conso}

The notion of consonance is fundamental in multiple testing.  A multiple testing rule is said to be consonant when, for any set $A$, if the global null $H_A$ is rejected, there is at least one singleton $i \in A$ such that the individual null $H_i$ is also rejected  \citep{gabriel1969simultaneous}. Formally, if the multiple test has decision rule $\Phi$, then it is consonant if for any $A$, there is an index $i \in A$ such that  $ \Phi (p_{A}) \le \Phi (p_i) $. 

Consonance also leads to computationally efficient closed testing rules under some conditions.  \cite{hommel2007powerful} show that consonance leads to a shortcut of order $n$ for closed testing. However, their shortcut is only feasible if one can identify the elementary hypothesis in an efficient way. Unfortunately, this is only known for weighted Bonferroni tests.  \cite{hommel2007powerful} acknowledge this limitation, writing that ``\emph{after rejecting $H_A$ it may sometimes remain difficult to identify an
elementary hypothesis $H_i$, $i \in A$, to be rejected. In such cases the short-cut can still be computer
intensive}''. See also \cite{brannath2010shortcuts} for methods based on local consonance  for restricted hypotheses.

Going back to the main topic of the paper, suppose now that the local test statistics used in closed testing are also consonant. Assuming in addition that they are symmetric and monotone, as in the previous sections, it is easy to see that one obtains an efficient dynamic programming-type shortcut for computing the closed testing method.  Indeed, the closed testing decision rule for the smallest $p$-value reduces to 
$$\Phi_c(p_{(1)}) = \Phi(p_{(1)},\ldots, p_{(n)}).$$ 
To see this, first we notice that the decision rule clearly must include the above factor, i.e., the intersection null with all $n$ $p$-values must be rejected in order for the smallest $p$-value to be rejected. Next, by consonance, if the null with these $p$-values is rejected, then there must be a singleton $j$ that is rejected. Now, since $p_{(1)} \le p_j$, we obtain that $p_{(1)}$ is also rejected. This shows that the decision rule for the smallest $p$-value has the above form. 

With a similar reasoning, we obtain that the decision rule for the second $p$-value has the form 
$$\Phi_c(p_{(2)}) = \Phi(p_{(1)},\ldots, p_{(n)}) \cdot \Phi(p_{(2)},\ldots, p_{(n)}).$$

Therefore, we must only test the null with $p$-values $p_{(2)},\ldots, p_{(n)}$ to compute the decision rule for the second smallest $p$-value. 
Continuing, this shows that an algorithm of complexity $O(ns)$  exists for computing consonant closed tests based on monotone symmetric testing rules. For instance, the Bonferroni rule with $\Phi(p_{1},\ldots, p_{k}) = I(\min p_i \le \alpha/k)$ satisfies these properties. It is also well known that the Hommel procedure is a consonant closed testing procedure  \citep{sonnemann1982allgemeine,sonnemann2008general}. Therefore, consonance and closed testing lead to extremely fast algorithms. 

However, we emphasize that the scope of this paper goes \emph{much beyond} consonance. In this section, we will show here that there are important examples of closed tests based on monotone and symmetric rules that are \emph{not} consonant. We will establish conditions needed for consonance, and then show that specific tests do not satisfy them. 

For this we take a systematic approach. Suppose we are testing $n$ null hypotheses using the closure of monotone combination tests $T(p_J) = \sum_{i\in J}f(p_i)$. Recall that the global rules are $\Phi(p_J) = I(T(p_J) \le c_{|J|})$, where $c_{|J|} = c_{|J|,\alpha}$. We have the following key result, which clarifies the conditions on closures of monotone sums under which we have consonance. We call this result the Consonance-Closure-Monotonicity (CCM) lemma. See Section \ref{pf:ccm} for the proof.

\begin{lemma}[Consonance-Closure-Monotonicity (CCM) lemma]
\label{ccm}
Let $F_k$ be the cdf of an average of $k$ random variables $f(P_i)$, where $P_i$ are independent $p$-values uniformly distributed on $[0,1]$: 
$$F_k(c) = \Pr\left( k^{-1} \sum_{i=1}^k f(P_i) \le c\right).$$
Then the closure of monotone combination tests $\Phi(p_J) = I(T(p_J) \le c_{|J|})$ is consonant if and only if the following two conditions hold:
\benum
\item[(a)] {\bf Level}: Each local test has level $\alpha$: $$F_k(c_k) \le \alpha$$ for all $k \le n$.
\item[(b)] {\bf Sub-linear critical value growth}: The critical values $c_k$ grow at most linearly: $$c_k \le k c_1.$$
\eenum
\end{lemma}

From the CCM lemma, we can derive several important and interesting results. First, we study the consonance of two specific monotone combination rules, Stouffer's and the Truncated Product Method, and then we give a more general result.  See Section \ref{pf:cor-cons} for the proof. 
\begin{proposition}[Examples of Consonance]
\label{cor-conso}
\benum
\item
The closure of Stouffer's combination, where $T(p_J) = \sum_{i\in J} \Phi^{-1}(p_i)$, is not consonant if 
$$\alpha<1/2.$$
\item The closure of the Truncated Product Method, where $T(p_J) = \sum_{i\in J}\ln(p_i) I(p_i \le \tau)$, is consonant if
$$1-\sqrt{1-\alpha} \le \tau\le \alpha .$$ 
\item More generally, suppose that $f$ is a strictly increasing continuous function such that $\E|f(P)|<\infty$ for a uniform $p$-value $P$. Suppose we are testing $n$ null hypotheses. The closure of monotone combination tests $T(p_J) = \sum_{i\in J}f(p_i)$ based on independent $p$-values can \emph{only be consonant} for all $n$ if 
\beqs
\E f(P) \le f(\alpha).
\eeqs
\eenum
\end{proposition}

This result gives clear conditions under which closures of monotone sums are consonant, and gives specific examples. 
In particular, the third statement implies that the closure of Fisher's combination test is not consonant for all $n$ if $\alpha >1/e$. 

The condition above gives a lower bound on the critical value $\alpha$ for which the tests are consonant. The values of interest for us are typically $\alpha = 0.05$ or $\alpha = 0.01$. For these values, however, the classical combination tests (Fisher and Stouffer), are \emph{not} consonant. The closure of the Truncated Product Method is consonant for the choices of $\tau$ specified in Proposition \ref{cor-conso}. The FACT algorithm applies to the closure of Fisher's and Stouffer's test, neither of which are consonant in general. This shows that our FACT algorithm has a  broader scope than consonance. 

A key idea about consonance is \emph{consonantization} \citep{romano2011consonance}, which shows that under some conditions, any closed testing method can be replaced with a consonant one. Specifically, following the above reference, a family $\{H_i \}_{i=1,\ldots,n}$ of hypotheses is called elementary if there is no $i \neq j$ such that $H_i \subset H_j$.  Any closed testing method of elementary hypotheses can be made consonant, in the following way. Suppose $H_K$ is rejected when $\Phi_K=1$. Define 
$$\Phi_K' = \max_{i:i\in K} \prod_{J:i\in J}\Phi_J.$$
Then, the closure of $\Phi_K$ and $\Phi_K'$ reaches the same decisions about $H_i$, and the closure of $\Phi_K'$ is consonant. 

While this is an intriguing idea, it is not clear to us if it can be computed efficiently in a general context. This is an interesting problem that falls beyond our current scope.

\section{Addressing reader's concerns}
\label{conc}

Readers of this paper may have several concerns about the practicality of our methods. We attempt to address them below: 

\benum 
\item Is family-wise error rate (FWER) control feasible for large datasets, or is it too stringent? 

In the introduction, we give several references to important problems where FWER control is the state of the art. In particular, in the top genome-wide association studies (GWAS) published in \emph{Nature} and \emph{Science}, scientists still control the FWER! Thus, FWER is still extremely important. The interested readers may take a look at the following leading GWAS papers: \cite{wellcome2007genome,international2009common}. They all control FWER, and not FDR. See also the review paper in \emph{Nature Genetics} \citep{sham2014statistical}, where FWER control is the gold standard, and false discovery rate (FDR) control is presented as merely a possibility "\emph{rarely performed}". 
 
However, we agree that  family-wise error rate control can be too stringent for ultra-high dimensional data. FACT is not designed for such applications. 

\item What is the price for assuming monotonicity and symmetry? 

Most practical examples of closed testing (Holm, Bonferroni etc) are already symmetric and monotone, so FACT is more general than them. So, from the point of view of using more restricted methods, in most cases, we do not need to restrict the methods that we use. From a computational perspective, however, if we directly use the FACT algorithm, it can be slower than the specific shortcuts. Therefore, in each specific case, we should still use the specific shortcuts.

There are also examples of closed testing that are not symmetric, especially the graphical approaches of \cite{bretz2009graphical}, and other methods popular in clinical trials (gatekeeping approaches etc). It is possible, though beyond our scope, to conduct some simulation studies comparing FACT with those methods.

\item How can we choose symmetric and monotone tests? And how much difference does it make in practice?

There are multiple ways of choosing tests, but that is a strength of the method, and not a weakness. It gives flexibility to the user. We outline one specific approach, the Simes-HC fusion, which is appropriate for the very important and common problem of sparse signals. There are always multiple ways to construct test statistics, and the optimality theory for multiple testing is not that well-developed. So I think that it would be quite difficult to choose optimal tests theoretically. 

However, it would certainly be possible to conduct simulation studies comparing different test statistics. Regarding how much difference it makes in practice, we think it is not possible to answer all questions in one paper, due to the space limitations. But this can certainly be investigated in follow-up work.  

\item How can we guess the sparsity level $s$?

We agree that guessing $s$ is not trivial. However, in the simulations we have found that the method is quite robust to the choice of $s$. This question deserves further attention, and it can certainly be investigated in simulations. 

Another relatedd question is if the
adaptive Simes-HC rule provably dominates both Simes algorithm and HC procedure? This is clearly true in simulations, but it would be valuable to have theoretical results.

\item How can we perform accurate $p$-value calculations for the HC local tests?

Obtaining accurate $p$-values for HC is a problem of ongoing research interest. For accurate asymptotics, see \cite{li2015higher}. For finite-sample problems, see \cite{barnett2014analytical}. For correlated $p$-values, see for instance the pioneering work by \cite{hall2010innovated}, and follow-ups.
\eenum 

\section{Proofs}

\subsection{Review of the well-known argument that CT controls the FWER}
\label{fw}
This is the well-known argument that CT controls the FWER. let $I_0$ be the set of all true nulls. Then, for all $i\in I_0$: 
$$\Phi_c(p_i)
 = \prod_{J: i\in J }\Phi(p_J)
 \le \Phi(p_{I_0}),$$
where in the first step we have used the definition of CT (reject individual null if all sets containing it are rejected), and that all rejection rules take values in $\{0,1\}$, while in the second step we have used that $I_0$ is among the sets $J$ containing $i$. We conclude that if $H_i$ is rejected, then so is $H_{I_0}$. Since the probability of rejecting $H_{I_0}$ is at most $\alpha$, this shows that CT controls the FWER. 
\subsection{Proof of Main Theorem}

Consider the closed testing method. To decide whether or not we reject the $k$-th hypothesis, we must decide for every subset containing $k$ whether or not it is rejected based on the local testing rule. Now consider subsets of a fixed size $j$ containing $k$. By assumption, for each of these subsets, we use the same local testing rule $T_j$. Also by assumption, these testing rules are symmetric and monotone. It follows that all subsets are rejected if and only if the ``worst'' one is rejected. The ``worst one'' has the largest $j - 1$ $p$-values excluding $p_k$. 

We can formalize this intuition as follows. Let $p^{-k}_{(1)} \le \ldots \le p^{-k}_{(n)}$ be the sorted $p$-values excluding $p_k$. Recall that the decision rule $\Phi_c$ for the $k$-th null is
$$\Phi_c(p_k) = \prod_{J: k\in J}\Phi(p_J).$$
We can write this as a product over subsets of each possible size $j$. By the above discussion, the $j$-th term equals $\Phi(p_k, p^{-k}_{(n-j+1)},\ldots, p^{-k}_{(n)})$. Therefore, the entire decision rule for the $j$-th term has the form

$$\Phi_c(p_k) = \Phi(p_k) \cdot \Phi(p_k, p^{-k}_{(n)}) \cdot \ldots \cdot \Phi(p_k, p^{-k}_{(1)},\ldots, p^{-k}_{(n)}).$$

Next, we observe that if $p_k$ is  less than or equal to $p_l$, then the term $\Phi(p_k, p^{-k}_{(a)},\ldots, p^{-k}_{(n)})$ is less than or equal to $\Phi(p_l, p^{-l}_{(a)},\ldots, p^{-l}_{(n)})$. This simply means that if we reject all subsets of a fixed size containing the $l$-th hypothesis, then we also reject all subsets of the same size containing the $k$-th hypothesis. Therefore, if we do not reject the ``worst'' subset of size $j$ (say) for $p_{(k)}$, we do not reject the ``worst'' subset for any $p_{(l)}$ with $l \ge k$. 

This shows that the following algorithm is equivalent to the closed testing method. Start by sorting the $p$-values. For each $p$-value $p_{(k)}$ starting with the smallest one, check if the ``worst'' subset of each size $j$ is rejected using the local testing rule $T_j$. If any such subset is not rejected, stop, and reject the hypotheses with the $p$-values $p_{(1)}, \ldots, p_{(k-1)}$. This agrees with the FACT algorithm, showing its correctness. 

Finally, we study the computational cost of the algorithm. The initial sort takes $O(n \log(n))$ steps. Then at step $k$, the cost is at most $$t_1+t_2+\ldots+t_{n-k+1},$$ where $t_i$ is the cost of applying $T_i$ to a size $i$ subset. For the total cost, suppose we reject $k$ out of $n$ hypotheses. Then the total cost is at most 
 
 $$C_k = O\left(n\log(n)+k\left[t_1+t_2+\ldots+t_{n-k}\right]+\sum_{j<k} jt_{n-j+1}\right) $$
 
  
For instance, if $t_n = O(n^c)$ with $c\ge 1$, then we get a total cost $O(kn^{c+1})$. When $t_i = O(i)$, i.e., when applying the test takes linear time, the cost is $C_k = O(kn^2) $. This finishes the proof.

\subsection{Proof of Lemma "Constructing monotone symmetric test statistics"}
\label{pf:cl_lem}

We first study the properties of the rejection regions of monotone symmetric sets. Then, we apply these results to prove the current result. Consider a monotone symmetric rule $\Phi:[0,1]^n\to\{0,1\}$. Let $A$ be the rejection region,  $A = \{p:\Phi(p) = 1 \}$. Note that $A \subset [0,1]^n$  must of course be Borel measurable. Then monotonicity and symmetry of $\Phi$ are equivalent to the same properties of $A$:

\benum
\item {\bf Set Monotonicity}: $p\in A$ implies $q\in A$, if on all coordinates $i$, $q_i \le p_i$. 
\item {\bf Set Symmetry}: $p = (p_1,\ldots,p_n) \in A$ implies $p_\pi = $ $(p_{\pi(1)},\ldots,$ $p_{\pi(n)}) \in A$ for any permutation $\pi$ of $\{1,\ldots,n\}$. 
\eenum

Let $\A$ be the collection of all such sets. What properties does $\A$ have?

\begin{lemma}[Union]
\label{union_lem}
The union of any collection of monotone symmetric sets is monotone symmetric.
\end{lemma}

\begin{proof}
Let $A_i$, $i \in I$, be a collection of monotone symmetric sets, and suppose $x\in \cup_i A_i$. Let any $y \le x$, where inequality is meant coordinate-wise, i.e., $y_i \le x_i$ for all $i$. Now, we must have $x\in A_i$ for some $i\in I$. Then, since $A_i$ is monotone, we have $y\in A_i$, hence $y \in \cup_i A_i$. Moreover, since $A_i$ is symmetric, we have $x_\pi \in A_i \subset \cup_i A_i$ finishing the proof.
\end{proof}

\begin{lemma}[Intersection]
\label{int_lem}
The intersection of any collection of monotone symmetric sets is monotone symmetric.
\end{lemma}

The proof is similar to the previous one, and hence omitted. 

For any $a \in [0,1]^n$, let us denote the hyper-rectangle with opposite vertices $0$ and $a$ by $H(a)$. Thus $$H(a) = \{x: 0\le x_i\le a_i, i=1,\ldots,n\}.$$ Clearly, $H(a)$ is a monotone set. Let moreover $H_\pi(a)$ be the symmetrization of $H(a)$, that is $$H_\pi(a) =  \cup_{\pi \in S_n} H(a_\pi).$$ Thus $H_\pi(a)$ is monotone and symmetric. Moreover, $H_\pi(a)$ is a minimal monotone symmetric set, in the sense that if $a$ belongs to a monotone symmetric set, then $H_\pi(a)$ must also belong to it. 

\begin{lemma}[Representation]
A set $A$ is monotone symmetric if and only if it can be written as $$A = \cup_{a \in I} H_\pi(a),$$ for some measurable set $I \subset [0,1]^n$.
\end{lemma}

\begin{proof}
Clearly, if this representation holds, then by the union lemma, Lemma \ref{union_lem}, $A$ is monotone symmetric. On the other hand, if $A$ is monotone symmetric, then clearly $A$ has the given representation, with $I=A$. This finishes the argument. 
\end{proof}

\subsection{Final proof of Lemma "Constructing monotone symmetric test statistics"}
Clearly,  a test statistic $T$ is monotone and symmetric if and only if all its sub-level sets $S_T(c) = I(T(p) \le c)$ are monotone symmetric in the sense of sets, as defined above. 

Now, for the first part, let $T=\min_{i \in I} T^i$ be a minimum of an arbitrary measurable collection of monotone symmetric test statistics. We notice that $\min_{i \in I} T^i \le c$ iff for some $i\in I$ we have $T^i \le c$. Thus $$S_T(c)  = \cup _{i \in I} S_{T^i}(c).$$ Thus, $T$ is monotone symmetric the union lemma, Lemma \ref{union_lem}. 

Similarly, for the second part let $T=\max_{i \in I} T^i$ be a maximum of an arbitrary measurable collection of monotone symmetric test statistics. Clearly, $\max_{i \in I} T^i \le c$ iff for all $i\in I$ we have $T^i \le c$. Thus $S_T(c)  = \cap _{i \in I} S_{T^i}(c) $. Thus, this property follows from the intersection lemma, Lemma \ref{int_lem}.

For the third part, it is clear from the original definition for test statistics that monotonicity and symmetry are preserved under non-negative combinations. 

For the fourth part, suppose that on all coordinates $i$, $p_i \le p'_i$. Then, by monotonicity of the $T^j$, $T^j(p) \le T^j(p')$ for all $j$, so that by monotonicity of $g$, $$g(T^1(p),\ldots,T^k(p)) \le g(T^1(p'),\ldots,T^k(p')).$$ This shows that $g(T^1,\ldots,T^k)$ is a monotone test statistic. Next, to see the symmetry, we notice that the value of each $T^j$ is unchanged under permutations. Therefore, the value of $g$ is also unchanged, finishing the proof.



\subsection{Proof of Proposition \ref{fct_bonf}}
\label{pf:fct_bonf}

In general, the decision rule for the hypothesis corresponding to the first $p$-value equals
$$\Phi(p_{(1)}) \cdot \Phi(p_{(1)}, p_{(n)}) \cdot \ldots \cdot \Phi(p_{(1)}, p_{(2)},\ldots, p_{(n)}).$$

Since we are working with the Bonferroni method, the $i$-th factor equals $I(p_{(1)} \le \alpha/i)$. Since the thresholds are decreasing in $i$, the last rule is the most stringent one. Thus, we obtain the simplification that the decision rule for the smallest $p$-value equals
$$I(p_{(1)} \le \alpha/n).$$

If we have rejected the hypothesis corresponding to the smallest $p$-value, we continue and examine the second smallest one. This is rejected based on the rule 
$$\Phi(p_{(2)}) \cdot \Phi(p_{(2)}, p_{(n)}) \cdot \Phi(p_{(2)},p_{(n-1)}, p_{(n)})  \ldots \Phi(p_{(2)}, p_{(3)},\ldots, p_{(n)}). $$

As above, the $i$-th factor equals $I(p_{(2)} \le \alpha/i)$. Since the thresholds are decreasing in $i$, the last rule is the most stringent one. Thus, we obtain that the second decision is based on
$$I(p_{(2)} \le \alpha/(n-1)).$$
Continuing similarly, we obtain that the decision to reject the hypothesis corresponding to $p_{(i)}$ is only considered if the hypotheses corresponding to the smaller $p$-values are all rejected. The hypothesis is then rejected if $p_{(2)} \le \alpha/(n-i+1)$, which agrees with Holm's method. This shows that the steps of the FACT algorithm recover Holm's method.

\subsection{Proof of Proposition \ref{fct_simes}}
\label{pf:fct_simes}

The reasoning proceeds from the last loop of the FACT algorithm sequentially towards the first loop. We show that the decision made at each loop matches the decision at the corresponding step of Hommel's procedure. We view Hommel's procedure as a sequential algorithm proceeding from $j=0$ to $j=n$. At the $j$-th step, we check if $p_{(n-j+k)}> k\alpha/j$ for all $k \le j$. If this is the largest $j$ with this property, we reject all $H_i$ with $p_i \le \alpha/j$. 

Thus, consider the last outer loop of the FACT algorithm, where $k=n$. There, if $p_{(n)} \le \alpha$, all hypotheses are rejected. This matches the first step in Hommel's procedure, where $j=0$. Indeed, in that case, there is no $j$ such that $p_{(n-j+k)}> k\alpha/j$ for all $k \le j$. This means, for $j=1$, that $p_{n}\le \alpha$. 

However, we also need to argue that if $p_{(n)} \le \alpha$, then the FACT algorithm indeed arrives at the $n$-th loop, and does not stop before that. This is indeed true, because if  $p_{(n)} \le \alpha$, then each of the conditions in each of the previous $n-1$ loops are fulfilled. This shows that the first steps of the two algorithms agree.

It remains to understand the case where $p_{(n)}> \alpha$. In this case, the FACT algorithm stops before the $n$-th outer loop. As before, there are two conditions for it to stop at the $n-1$-st outer loop. First, both of the conditions 
$$\{ p_{(n-1)} \le \alpha/2 \} \cup \{ p_{(n)} \le \alpha \}$$ 
and $\{ p_{(n-1)} \le \alpha \}$ must be satisfied. Since $p_{(n)}> \alpha$, this means that the required condition is $ p_{(n-1)} \le \alpha/2$.  The second condition is that the algorithm must not stop before the $n-1$-st outer loop. However, we claim that the condition $ p_{(n-1)} \le \alpha/2$ guarantees that. 

Indeed, consider a previous outer loop, say the $k$-th one with $k\le n-2$. Then, we claim that for any inner loop $j$, there is a $p$-value satisfying the Simes constraint. Indeed, consider first the loops $j$ for which $ p_{(n-1)}$ is one of the $p$-values considered in $p_{(k)},p_{(j+1)},\ldots,p_{(n)}$. Thus, $j+1 \le n-1$, or $j\le n-2$. In this case, the threshold to which $ p_{(n-1)}$ is compared in the Simes test is $\alpha \cdot (n-j)/(n-j+1)$. Since this is greater than $\alpha/2$, we have that $$p_{(n-1)}\le\alpha \cdot (n-j)/(n-j+1),$$ and so the Simes constraint is satisfied in this case. 

Consider next the loops $j$ for which $j\ge n-1$. In this case, the threshold to which $ p_{(1)}$ is compared in the Simes test is $\alpha /(n-j+1)$. Since this is greater than or equal to $\alpha/2$, and $p_{(1)} \le p_{(n-1)}$ we have that the Simes constraint is satisfied in this case.

In conclusion, we have shown that the FACT algorithm stops at the $n-1$-st loop precisely when $ p_{(n)} > \alpha$ and $ p_{(n-1)} \le \alpha/2$. The algorithm stops before this loop when $ p_{(n-1)} > \alpha/2$. This agrees with the second step in Hommel's algorithm. 

By a similar inductive argument, we obtain that the FACT algorithm stops at the $n-j+1$-st loop precisely when 

$$p_{(n)} > \alpha, p_{(n-1)} > \alpha\cdot (j-1)/j, \ldots, \textnormal{and } p_{(n-j+1)} \le \alpha/j.$$ 

This agrees with the $j$-th step in Hommel's algorithm, and finishes the proof.

\subsection{Proof of Proposition "GST-Sym-Mon"}
\label{pf:gst}

Let us write $d_i = d_{|J|,i}$ for simplicity. Define the regions $$S_i(d_i)  = \{p: p_{(i)} \le d_{i} \}.$$
The acceptance region of the Generalized Simes Test equals $S=\cup_i S_i(d_i)$. The order statistics are monotone symmetricy, hence the sets $S_i(d_i)$ are monotone symmetric. So by the closure properties (CCM Lemma) of monotone symmetric sets, $S$ is monotone symmetric, finishing the proof.

\subsection{Proof of CCM Lemma}
\label{pf:ccm}

For a test of the form $\Phi(p_J) = I(T(p_J) \le c_{|J|,\alpha})$, consonance requires that if $T(p_J) \le c_{|J|,\alpha}$, then there is an index $i \in J$ such that $T(p_i) \le c_{1,\alpha}$. For monotone combination tests $T(p_J) = \sum_{i\in J} f(p_i)$ with $k = |J|$, consonance requires that if $$\sum_{i\in J} f(p_i) \le c_{k,\alpha},$$ then $f(p_{(1)}) \le c_{1,\alpha}$. 

Now, in the worst case, we can take the $p$-values to be equal. Thus, consonance can only hold for all $p$-values if 
$\frac{c_{k,\alpha}}{k}\le c_{1,\alpha}.$
This finishes the proof. 

\subsection{Proof of Proposition \ref{cor-conso}}
\label{pf:cor-cons}
\benum
\item 
It is easy to see that the condition $c_{k,\alpha}\le k c_{1,\alpha}$ does not hold for Stouffer's test. We have $$c_{k,\alpha} = k^{1/2} \Phi^{-1}(\alpha),$$ thus $$c_{k,\alpha}/k = k^{-1/2} \cdot c_{1,\alpha}.$$ Since $c_{1,\alpha} = \Phi^{-1}(\alpha)<0$ if $\alpha <0$, we thus have $c_{k,\alpha}/k > c_{1,\alpha}$.  This finishes the proof of this claim. 

\item For the truncated product method, we have $f(x) = \ln(x) I(x \le \tau)$. For the critical value for one test, $c_1 = c_{1,\alpha}$, we need that 
$$F(c_1) = \Pr\left( \ln(P) I(P \le \tau) \le c_1 \right) \le \alpha.$$
Now, taking $c_1 = \ln  \tau$, we see that this probability equals $\tau$. Since the value of the random variable $\ln(P) I(P \le \tau)$ equals zero for $P>\tau$, it follows that  $F(c_1) = \tau$ for any $c_1 \in (\ln  \tau,0)$. Since we are interested in the regime where $\tau \le \alpha$, it follows that we can take any $c_1 \in (\ln  \tau,0)$ and the level condition is satisfied for subsets $J$ of size one. 

Now consider subsets of size two. Let $X_i = \ln(P_i) I(P_i \le \tau)$ and $c_2 = c_{2,\alpha}$ be the appropriate critical value. We need that 
$$F_2(c_2) = \Pr\left( X_1+X_2  \le c_2 \right) \le \alpha.$$
Since $X_i\le 0$, we have that $X_1+X_2=0$ only if $X_1=0$ and $X_2=0$. The probability of this event is $(1-\tau)^2$. Moreover, the next largest value that $X_1+X_2$ can take with positive probability equals $\ln \tau$, which happens when one $X_i=0$ and the other $X_j = \ln \tau$. Therefore, we conclude that 
$$F_2(c_2) = 1 - (1-\tau)^2$$
for all $c_2 \in (\ln\tau,0)$. For this to be at most $\alpha$, we need precisely that $1-\sqrt{1-\alpha} \le \tau$, which is the required condition. 

For consonance, it remains to show that one can choose $c_2$ such that $c_2 \le 2c_1$. This is clear, because $c_1,c_2$ are only constrained to be in $(\ln\tau,0)$.

Finally, consider subsets of size $k$. Similarly to above, we derive that 
$$F_k(c_k) = 1 - (1-\tau)^k$$ for all $c_k \in (\ln\tau,0)$. Thus, the level requirement translates to $$1-(1-\alpha)^{1/k} \le \tau.$$ Since $x \to x^{1/k}$ is increasing for $x\in(0,1)$, this condition is implied by the one for $k=2$. Moreover, for consonance, we need $c_k \le kc_1$. This holds similarly to the case $k=2$, finishing the proof.

\item 
One can check that $ c_{1,\alpha} = f(\alpha)$ is a valid choice of a critical value for strictly increasing continuous $f$. Indeed, $\Pr(f(P_i)\le f(\alpha))=\Pr(P_i\le \alpha)= \alpha$. Thus, consonance requires equivalently that 
$$p_n := \Pr\left(\sum_{i=1}^n f(p_i) \le n f(\alpha)\right) \ge \alpha.$$
By the law of large numbers, $n^{-1}\sum_{i=1}^n f(p_i) \to \E f(P)$ almost surely, thus if $\E f(P) > f(\alpha)$, then $\lim\sup_n p_n =0$. Therefore, consonance for all $n$ requires that $\E f(P) \le f(\alpha)$, which finishes the proof. 
\eenum

\section{Numerical experiments}
\label{sec:numerical}

We perform numerical simulations to understand and compare the behavior of our proposed methods.

\subsection{Fusion rules are more powerful}

\begin{figure}
\begin{subfigure}{.32\textwidth}
  \centering
  \includegraphics[scale=0.35]{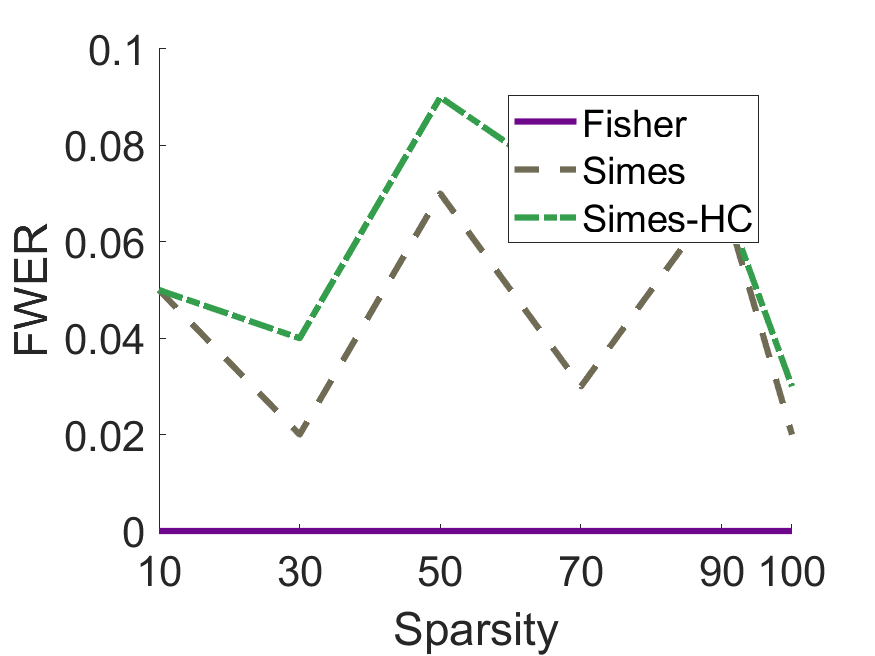}
  \caption{$M =  0$}
\end{subfigure}
\begin{subfigure}{.32\textwidth}
  \centering
  \includegraphics[scale=0.35]{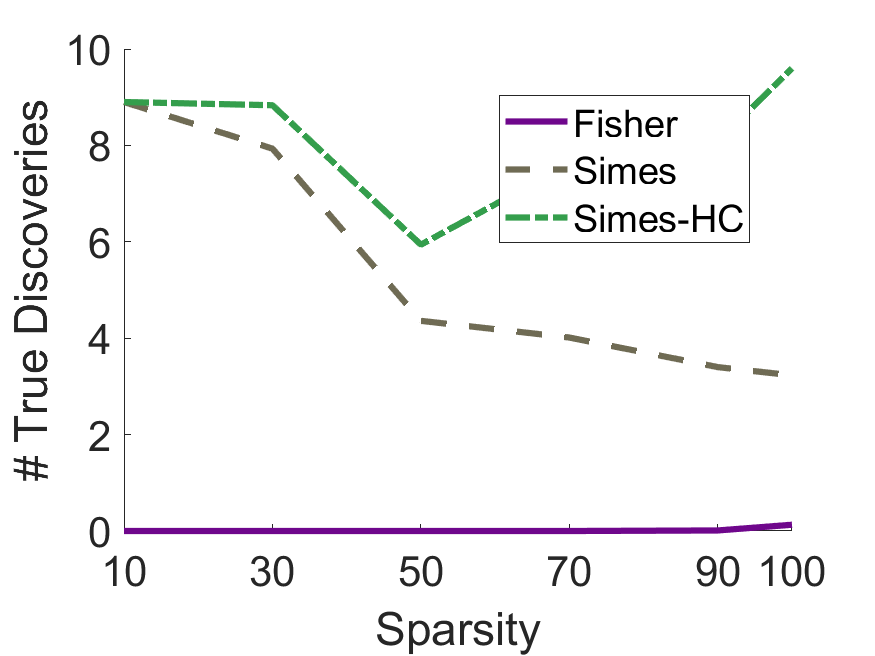}
  \caption{$M =  1$}
\end{subfigure}
\begin{subfigure}{.32\textwidth}
  \centering
  \includegraphics[scale=0.35]{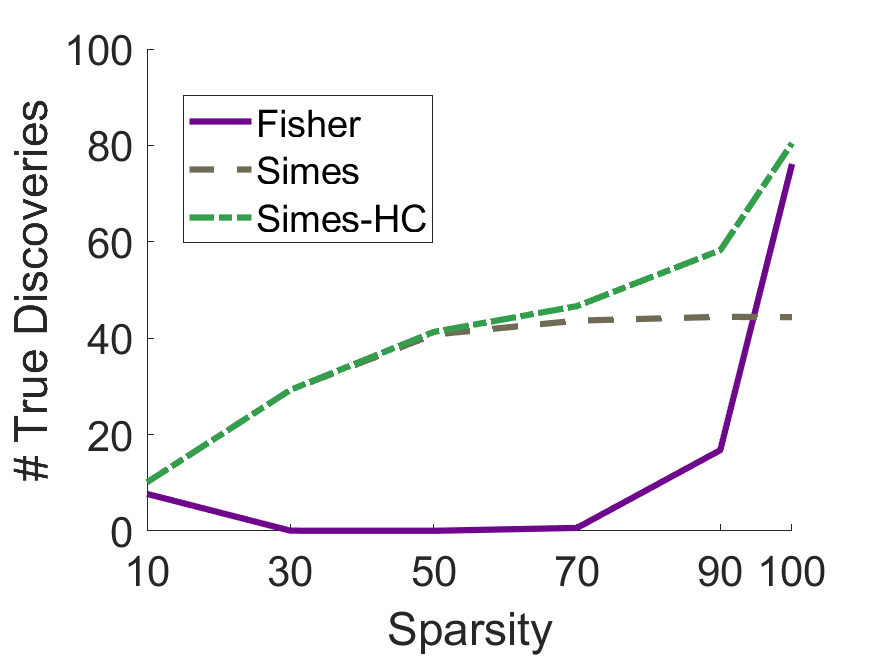}
  \caption{$M =  2$}
\end{subfigure}
\caption{FWER and average number of hypotheses rejected by the FACT method for the closure of Fisher's method, Simes' method, and the Simes-HC fusion rule. The displays are as a function of the signal strength $M$ and sparsity $s$.}
\label{fig:1}
\end{figure}

We compare the closure of Fisher's method, Simes' method, and FACT with the Simes-Higher Criticism (Simes-HC) fusion rule. The reason for performing this simulation is that  we would like to understand under what condition the new fusion rule can perform better than Simes method or Fisher's method.  We do not compare Holm's method, because the closure of Simes' method is more powerful. 
 
In the simulation, we use the normal means model,  where the data $X_i$ is independent and normally distributed with $X_i \sim \N(\mu_i,1)$ for $i=1, \ldots ,100$. The null hypotheses considered are that $\mu_i=0$. The alternative hypotheses are that $\mu_i>0$.

We consider both sparse and dense models. We change the sparsity---the number of nonzero $\mu_i$ effect sizes---of the model on a grid from zero to 100. We tune the effect sizes so that the power is in a non-trivial regime, and comparable across the different sparsities. Specifically, it is well known that the chi-squared  test for a global null based on $k$ normal means behaves approximately as $$\N(k+|\mu|_2^2, 2k)$$  for large $k$, where $|\mu|_2$ is the Euclidean norm of the vector of means. Therefore we choose our effect sizes in the following way. For a  global effect size $M$, and sparsity $s$, we set each nonzero effect size to be equal to $\mu_i = (2p/s)^{1/2}M$. This ensures that the Euclidean norm of the effects is the same for different sparsities. We take the signal strength $M$ to be 0, 1, and 2, respectively. 

We then run the Fast Closed Testing method using various local testing methods. We always use the global significance level $\alpha = 0.05$.  
We average the results over 100 independent Monte Carlo trials. 

We compare the closure of Fisher's method, Simes' method, and the FACT algorithm using the Simes-Higher Criticism (Simes-HC) fusion rule. For the Simes-HC rule, we assume that the preliminary estimate of the sparsity is correct. The results are displayed in Figure \ref{fig:1}. We observe that all methods empirically control the family-wise error rate, up to random sampling error. This can be seen on the left plot showing the FWER, where $M=0$, so that we are under the global null. On the remaining plot, we show the expected number of true discoveries. 

We observe that the closure of the Simes-HC fusion rule is more powerful than the other methods. First, the closure of Fisher's method only has power against fully dense alternatives, and only when the effect sizes are large ($M=2$). Second, as seen in the middle plot, the closure of Simes' method can lose power when the sparsity increases.

Finally, FACT with the Simes-HC rule has more power than the other methods. This can be seen best on the rightmost plot, where Simes HC tracks Simes for sparse alternatives, and tracks Fisher for dense alternatives. This shows that the FACT method, when used with the appropriate local testing rules, can be more powerful that the closure of a minimum or sum-based local testing rule, and demonstrates the power of our approach.

\subsection{Robustness to mis-specifying the sparsity}

\begin{figure}
\begin{subfigure}{.32\textwidth}
  \centering
  \includegraphics[scale=0.35]{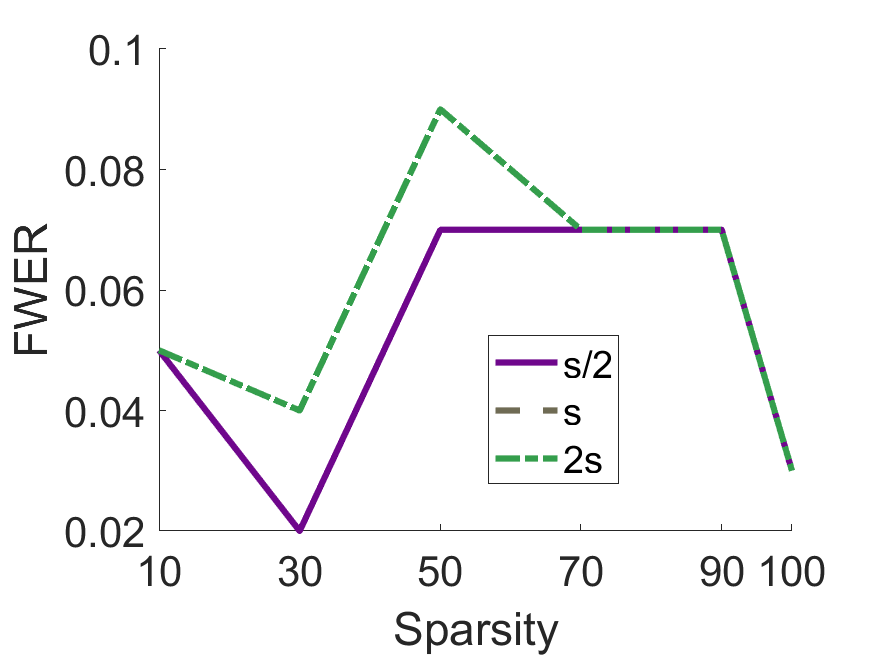}
  \caption{$M =  0$}
\end{subfigure}
\begin{subfigure}{.32\textwidth}
  \centering
  \includegraphics[scale=0.35]{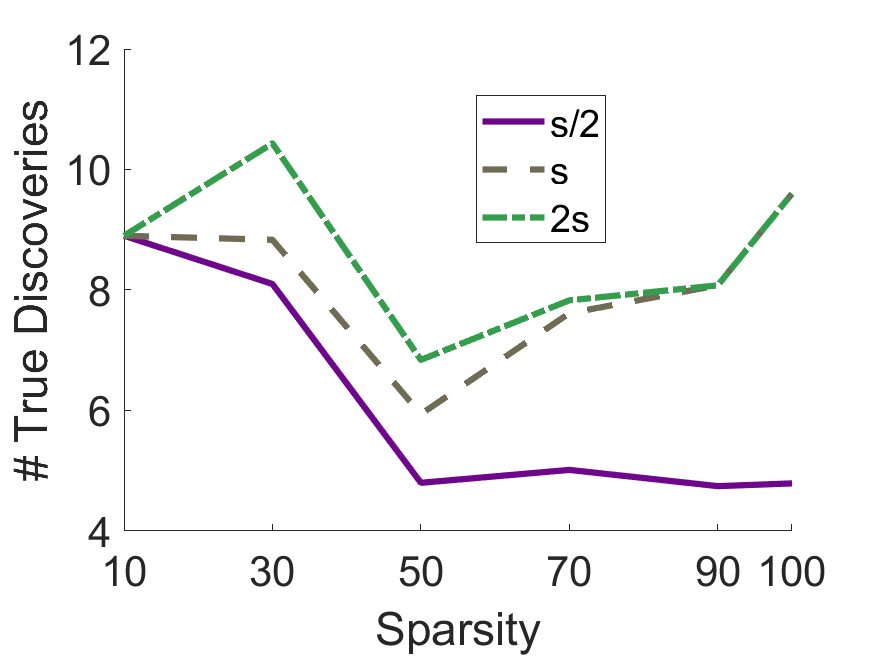}
  \caption{$M =  1$}
\end{subfigure}
\begin{subfigure}{.32\textwidth}
  \centering
  \includegraphics[scale=0.35]{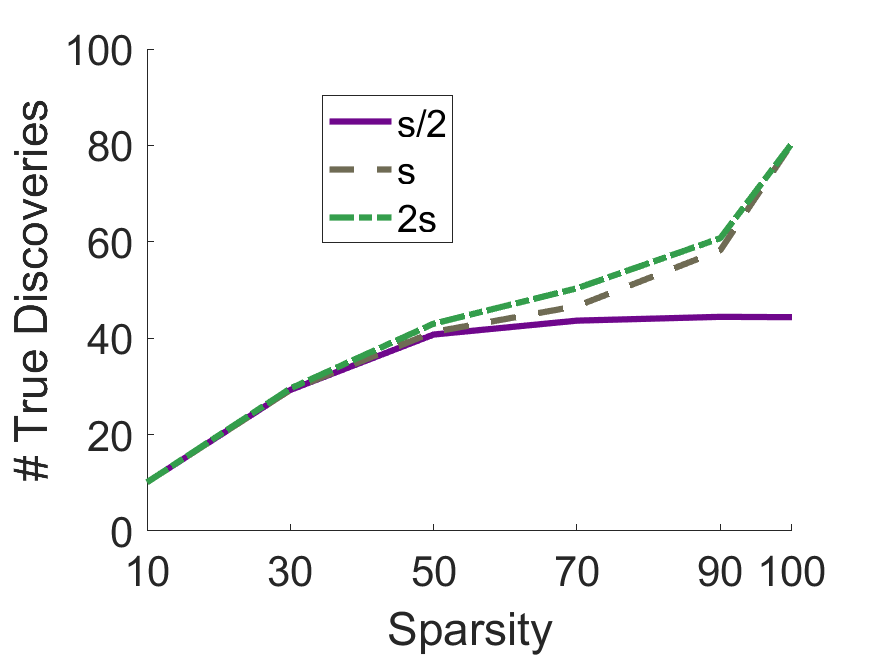}
  \caption{$M =  2$}
\end{subfigure}
\caption{FWER and average number of hypotheses rejected by the FACT method with the Simes-HC fusion rule, for potentially mis-specified sparsity levels. The displays are as a function of the signal strength $M$ and sparsity $s$.}
\label{fig:2}
\end{figure}

We examine the robustness of the Simes-HC fusion rule to the sparsity tuning parameter. The reason for performing this simulation is that we would like to understand how the performance of the method depends on our prior guess for the sparsity $s$. We use the same simulation setup as in the previous section. We now run the Simes-HC rule with sparsity tuning parameters equal to $s/2,s$, and $\min(2s,n)$. The results are displayed in Figure \ref{fig:2}.  

We observe that the method is not too sensitive to the sparsity tuning parameter.   Regardless of the value of that parameter, the Simes-HC rule always has good power. Underestimating the sparsity seems to lead to a smaller number of rejections. This suggests that in practice one should use over-estimates of the sparsity for robust performance. A more detailed investigation is beyond our current scope.

We also show that the FACT method to with the Simes-HC fusion rule is quite robust to correlated test statistics (Section \ref{sec:corr}).

\subsection{Robustness to correlated test statistics}\label{sec:corr}

\begin{figure}
\begin{subfigure}{.5\textwidth}
  \centering
  \includegraphics[scale=0.4]{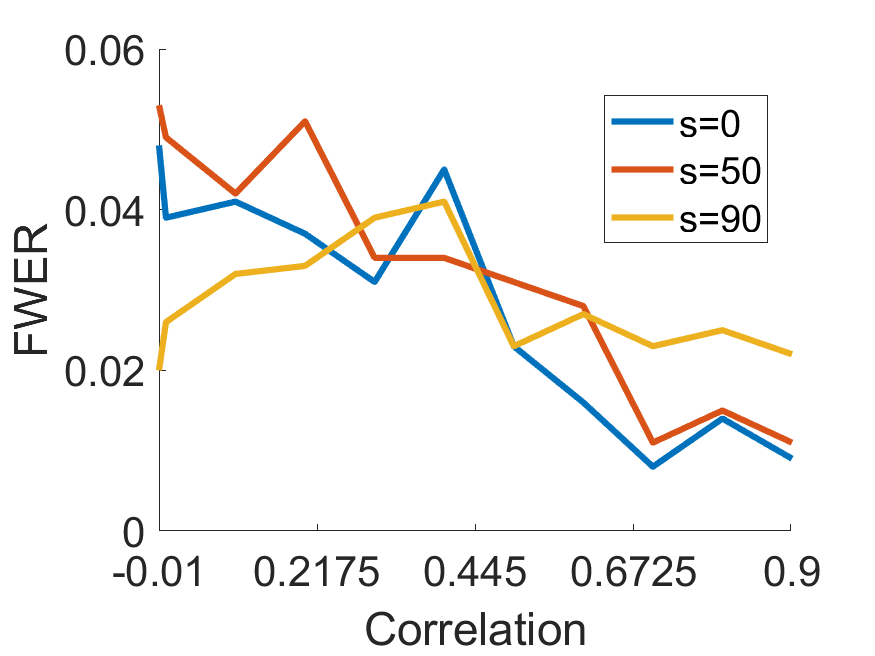}
  \caption{Spiked model}
\end{subfigure}
\begin{subfigure}{.5\textwidth}
  \centering
  \includegraphics[scale=0.4]{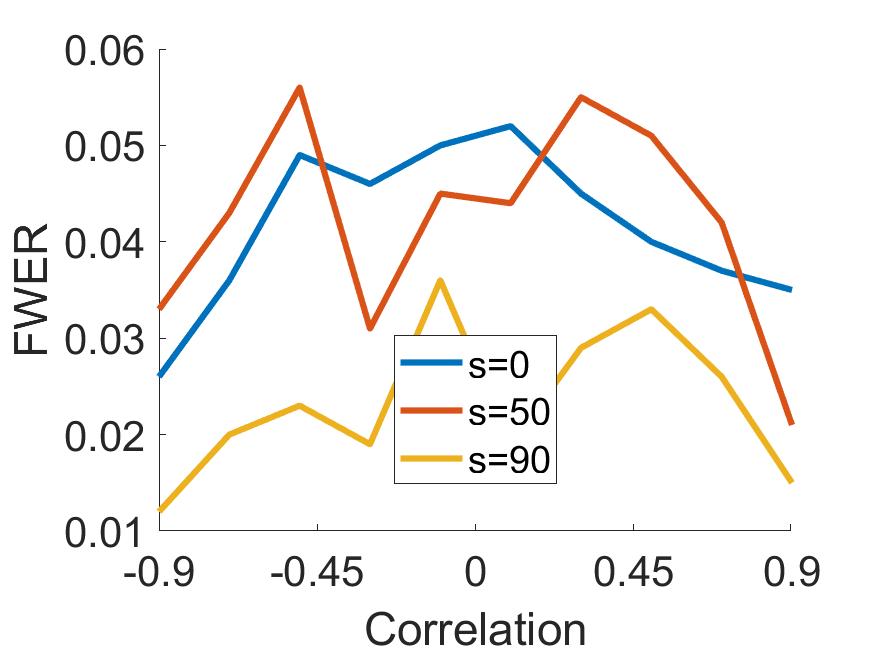}
  \caption{AR-1}
\end{subfigure}
\caption{FWER of the FACT method with the Simes-HC fusion rule, for correlated test statistics. The displays are as a function of the correlation coefficient.}
\label{fig:3}
\end{figure}

We examine the robustness of the FACT method to with the Simes-HC fusion rule to correlated test statistics. The reason is to understand how the FWER depends on the correlation structure of the tests.

We use the same simulation setup as in the previous section, with effect size $M=1$, and varying sparsity. Moreover, we sample the test statistics as
$$ X \sim \N(\mu,\Sigma),$$
where $\Sigma$ is a covariance matrix. We choose $\Sigma$ to be either a spiked covariance matrix $\Sigma = (1-\rho) I_n + \rho 11^\top$, where $1$ is the vector or all ones, or an autoregressive covariance matrix of order one (AR-1), $\Sigma_{ij} = \rho^{|i-j|}$. We let the correlation coefficient $\rho$ vary over the entire range where $\Sigma$ is non-negative definite. This includes both positive and negative correlation structures.  The results are displayed in Figure \ref{fig:3}. We average over 1000 Monte Carlo trials.  We observe that the method essentially controls the FWER for all correlation structures in this case.

\section{Data Analysis}
\label{g}

\begin{figure}
  \centering
  \includegraphics[scale=0.5]{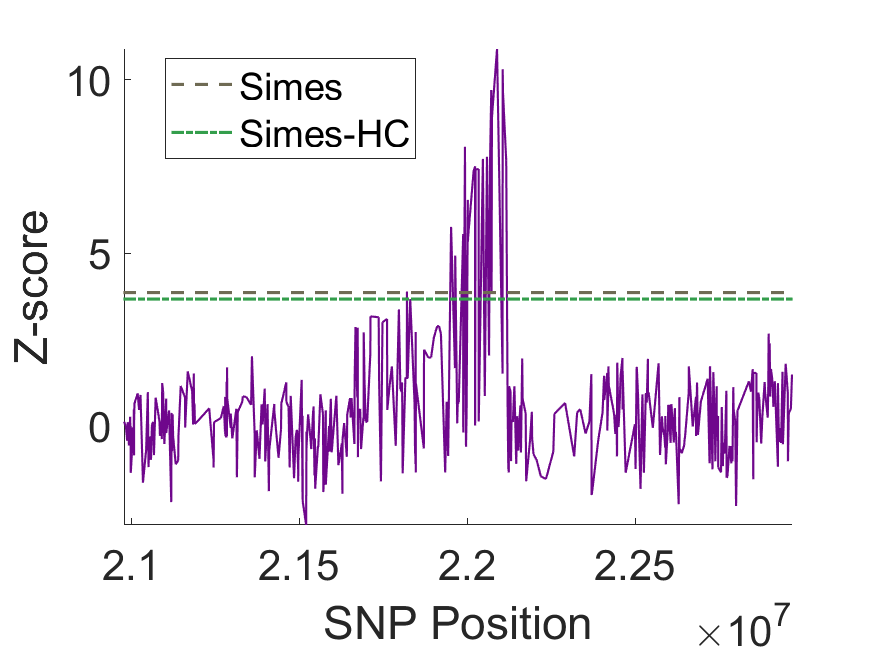}
\caption{Plot of Z-scores for association with coronary artery disease in the C4D  GWAS dataset, focusing on the 9p21.3 locus. The horizontal lines are the critical values of the closure of Simes' method (i.e., Hommel's method), and the FACT algorithm with the Simes-HC fusion rule. The Simes-HC method discovers a second associated SNP in the secondary cluster.}
\label{fig:h}
\end{figure}

We illustrate the FACT method on a Genome-Wide Association Study (GWAS) of coronary artery disease.  In the last decade, GWAS have become the backbone of modern medical genomics. Thousands of such studies have been performed, and have led to hundreds of novel associations between common traits and genetic variants \citep[see e.g.,][for a review]{visscher2012five}. Thus there is a great deal of interest in improving statistical inference via multiple testing in this area.

GWAS are a general and flexible type of genomic study to understand complex traits and diseases. In a typical GWAS, we collect a large number of cases and controls for a disease of interest, such as coronary artery disease. We also measure the genotypes of the samples for potentially hundreds of thousands of genetic variants known as Single Nucleotide Polymorphisms (SNPs). We then find the most significantly associated SNPs using multiple testing.

In this paper, we will llustrate the FACT method on the GWAS dataset from the C4D consortium for coronary artery disease genetics \citep{coronary2011genome}. This dataset contains about 500,000 genotyped SNPs, on  15,420 individuals with CAD (cases) of which 8,424 are Europeans and 6,996 are South Asians, along with 15,062 controls. Because out Simes-HC fusion method can more powerful than the closure of Simes' method (i.e., Hommel's method) for relatively dense alternatives, we focus on a subset of the SNPs that is already known to be enriched for CAD-related loci. Specifically, we focus on the neighborhood of the CDKN2A gene, at the 9p21.3 locus. 

This locus is known to be strongly associated with CAD from prior work, however the functional mechanism appears to be not fully known \citep[see e.g.,][]{harismendy20119p21, chen2014functional}. Therefore, it is of interest to better understand the local genetic architecture of this region. We focus on a  200 kilobase region centered at the position 21967752, which is the location of one of the most significant hits in the current study. There are $J = 452$ SNPs in this window.

We perform multiple testing with the closure of  Simes' method (i.e., Hommel's method) and the FACT method using the Simes-HC fusion. For the Simes-HC fusion, we set the  sparsity as $s = 0.1 J$. Hommel's method finds 24 significant loci, while the FACT with Simes-HC finds 25. The results displayed in Figure \ref{fig:h} show that most of the discoveries are in a contiguous window around the center of the region. 

Both Hommel's method and the fusion find significant discoveries in a second cluster, about 20 kilobases away from first one, but the fusion finds two loci there. This gives stronger evidence for the association of that cluster. Regarding interpretation, it is most likely that there are several loosely dependent SNPs associated with CAD in this region, so we think that it is valuable to have stronger evidence near the secondary locus.

Finally, we note that we expect the test statistics in this region to be dependent. However, we do not have access to the full dataset, but only to the list of $p$-values, and hence performing permutation methods is not feasible. Heuristically, the dependence should be positive, because of linkage disequilibrium between neighboring SNPs. Hence Simes' rule is expected to have the appropriate level. The same statement is less clear about the higher criticism, and this specific question deserves further study.

\section{Discussion of future work}
\label{disc}
In our work, symmetry was crucial for deriving efficient algorithms. However, there are important non-symmetric multiple testing methods, including weighted methods \citep{holm1979simple,dobriban2015optimal,fortney2015genome,dobriban2016weighted}, fixed-sequence procedures \citep{maurer1995multiple,westfall2001optimally}, and fallback procedures \citep{wiens2003fixed}. These control the FWER, because they are closures of weighted Bonferroni methods \citep{hommel2007powerful}. However, there is currently no general explanation for why they admit efficient algorithms. Can we derive such principles?  Can these principles lead to new multiple testing methods?

Second, closed testing methods have appealing power properties. It is known that any multiple testing method can be replaced by a closed testing method that rejects the same, and possibly more, hypotheses while controlling the FWER \citep{sonnemann1982allgemeine,sonnemann2008general,sonnemann1988vollstandigkeitssatze}. However, it is not known how to compute this closure efficiently. Are there conditions under which it can be computed in polynomial time?

In addition, it would be of major practical interest to study the case of \emph{dependent} $p$-values. 
 For this we would need a new probability inequality for local tests based on symmetrically distributed--or exchangeable--$p$-values, similar to Bonferroni's inequality for arbitrary dependence, and Simes' inequality for positive dependence \citep{goeman2014multiple}. Then we could apply the closure principle as in this paper. 

Moreover, there are many multiple testing methods that go beyond the closure principle. For instance, the it would be interesting to see when the methods based on the sequential testing principle \citep{goeman2010sequential} can be computed efficiently. Understanding these questions will help broaden our methods for multiple testing.

{\small
\setlength{\bibsep}{0.2pt plus 0.3ex}
\bibliographystyle{plainnat-abbrev}
\bibliography{references}
}
\end{document}